\documentclass{llncs}
\usepackage{fullpage}

\usepackage{epsfig}
\usepackage{gensymb}
\usepackage[english]{babel}
\usepackage{graphicx}
\usepackage{amssymb}
\usepackage{multirow}
\usepackage{subfigure}
\usepackage{hyperref}
\usepackage{xcolor}
\usepackage{xspace}
\usepackage{paralist}
\usepackage{amsmath,amscd,amssymb,german}
\usepackage{mathrsfs}
\usepackage{graphicx}
\usepackage[scaled]{helvet}
\usepackage{dsfont}
\usepackage[T1]{fontenc}
\usepackage[ansinew]{inputenc}
\usepackage{paralist}
\usepackage{wrapfig}

\newcommand{\bTWObesc}{{\sc b2pbesc}\xspace}
\renewcommand{\dots}{...}

\newcommand{\NP}{$\mathcal{NP}$\xspace}

\renewcommand{\NP}{$\mathcal{NP}$}

\newcommand{\NPC}{\mbox{\NP-complete}\xspace}

\newcommand{\NPHN}{\mbox{\NP-hardness}\xspace}

\definecolor{blue}{rgb}{0.274,0.392,0.666}
\definecolor{red}{rgb}{1,0.3,0.3}
\definecolor{green}{rgb}{0,0.588,0.509}

\newcommand{\red}[1]{{{#1\xspace}}}
\newcommand{\blue}[1]{{{#1\xspace}}}

\newcommand{\bookinstance}[1]{$G^{#1}({V_1^{#1}} \cup {V_2^{#1}}, E^{#1})$\xspace}

\newcommand{\Gint}{$G_{\cap}$\xspace}
\newcommand{\Gr}[1]{$\red{G^{#1}_1}$\xspace}
\newcommand{\Gb}[1]{$\blue{G^{#1}_2}$\xspace}

\newcommand{\tlp}{{\sc {\em $T$}-Level Planarity}\xspace}

\newcommand{\sefeinstance}[1]{$\langle\red{G^{#1}_1},\blue{G^{#1}_2}\rangle$\xspace}
\newcommand{\sefe}{{\sc sefe}\xspace}
\newcommand{\tuba}{$^\dag$} \newcommand{\rome}{$^{\diamond}$} \newcommand{\kit}{$^{\circ}$}

\newcommand{\cpllong}{Clustered Planarity with Linear Saturators\xspace}
\newcommand{\cpl}{{\sc cpls}\xspace}

\title{Intersection-Link Representations of Graphs}
\date{}
\author{Patrizio {Angelini\tuba}, Giordano {Da Lozzo\rome}, Giuseppe {Di Battista\rome}, \\
    Fabrizio {Frati\rome}, Maurizio {Patrignani\rome}, Ignaz {Rutter\kit}
	\institute{
    \tuba T\"ubingen University, Germany\\
    \email{angelini@informatik.uni-tuebingen.de}\\
	\rome~Roma Tre University, Italy\\
    \email{\{dalozzo,gdb,frati,patrigna\}@dia.uniroma3.it}\\
    \kit~Karlsruhe Institute of Technology, Germany\\
    \email{rutter@kit.edu}
}}

\newcommand{\remove}[1]{}

\let\doendproof\endproof
\renewcommand{\endproof}{\qed\doendproof}

\begin{document}
\pagestyle{plain}
\maketitle

\begin{abstract}
We consider drawings of graphs that contain dense subgraphs. We introduce \emph{intersection-link representations} for such graphs, in which each vertex $u$ is represented by a geometric object $R(u)$ and in which each edge $(u,v)$ is represented by the intersection between $R(u)$ and $R(v)$ if it belongs to a dense subgraph or by a curve connecting the boundaries of $R(u)$ and $R(v)$ otherwise. We study a notion of planarity, called {\sc Clique Planarity}, for intersection-link representations of graphs in which the dense subgraphs are cliques.
\end{abstract}

\section{Introduction}

In several applications there is the need to represent graphs that are globally
sparse but contain dense subgraphs. As an example, a social network is
often composed of communities, whose members are closely interlinked,
connected by a network of relationships that are much less dense.
The visualization of such networks poses challenges that are attracting the study of several researchers (see, e.g.,~\cite{brw-wnv-01,hb-vosn-05}). One frequent approach is to rely on clustering techniques to collapse dense subgraphs and then
represent only the links between clusters. However, this
has the drawback of hiding part of the graph structure. Another approach that has been explored is the use of hybrid drawing standards, where
different conventions are used to represent the dense and the sparse portions of
the graph: In the drawing standard introduced
in~\cite{bbdlpp-valg-11,hfm-dhvsn-07} each dense part is represented by an
adjacency matrix while two adjacent dense parts are connected by a curve.

In this paper we study \emph{intersection-link representations}, which are hybrid representations where in the dense parts of the graph the edges are represented by the intersection of geometric objects  ({\em intersection} representation)  and in the sparse parts the edges are represented by curves  ({\em link} representation). 

More formally and more specifically, we introduce the following problem. Suppose that a pair $(G,S)$ is given where $G$ is a graph and $S$ is a set of cliques that partition the vertex set of $G$. In an \emph{intersection-link} representation, vertices are represented by geometric objects that are translates of the same rectangle. Consider an edge $(u,v)$ and let $R(u)$ and $R(v)$ be
the rectangles representing $u$ and $v$, respectively. If $(u,v)$ is part of a
clique (\emph{intersection-edge}) we represent it by drawing $R(u)$ and $R(v)$
so that they intersect, else (\emph{link-edge}) we represent it by a curve
connecting $R(u)$ and $R(v)$. An example is
provided in Fig.~\ref{fig:model}. 

\begin{figure}
	\centering
		\includegraphics[page=1,width=0.35\textwidth]{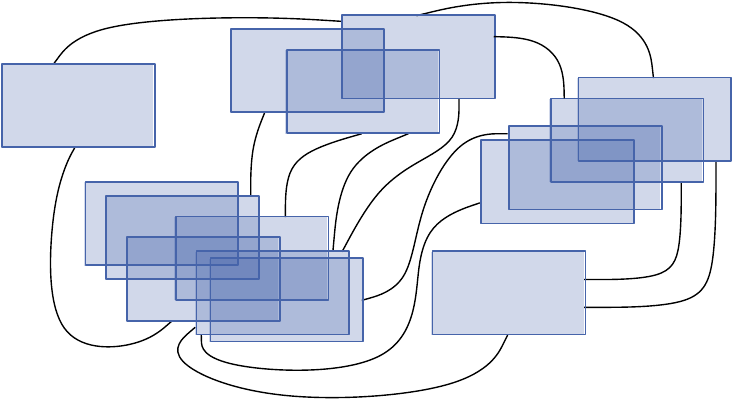}
	\caption{Intersection-link representation of a graph with five
cliques.}\label{fig:model}
\end{figure}

We study the {\sc Clique Planarity} problem
that asks to test whether a pair $(G,S)$ has an intersection-link representation such that link-edges do not cross each other and do not traverse any rectangle. The main challenge of the problem lies in the interplay between the geometric constraints imposed by the rectangle arrangements and the topological constraints imposed by the link edges.

Several problems are related to {\sc Clique Planarity}; here we mention two notable ones. The problem of recognizing intersection graphs of translates of the same rectangle is \NPC~\cite{f-sslg-09}. Note that this does not imply \NPHN for our problem, since cliques always have such a representation. {\em Map graphs} allow to represent graphs containing large cliques in a readable way; they are contact graphs of internally-disjoint connected regions of the plane, where the contact can be even a single point. The recognition of map graphs has been studied in~\cite{cgp-mg-02,t-mgp-98}. One can argue that there are graphs that admit a clique-planar representation, while not admitting any representation as a map graph, and vice versa.

We now describe our contribution. Our study encountered several interesting and at a first glance unrelated theoretical problems.  In more detail, our results are as follows.

\begin{compactitem}
\item In Section~\ref{se:hardness} we show that {\sc Clique Planarity} is \NPC even if $S$ contains just one clique with more than one vertex. This result is established by observing a relationship between {\sc Clique Planarity} and a natural constrained version of the {\sc Clustered Planarity} problem, in which we ask whether a path (rather than a tree as in the usual {\sc Clustered Planarity} problem) can be added to each cluster to make it connected while preserving clustered planarity; we prove this problem to be \NPC, a result which might be interesting in its own right.
\item In Section~\ref{se:fixed-representations}, we show how to decide {\sc Clique Planarity} in linear time in the case in which each clique has a prescribed geometric representation, via a reduction to the problem of testing planarity for a graph with a given partial representation.
\item In Section~\ref{se:2-cliques}, we concentrate on instances of {\sc Clique Planarity} composed of two cliques. While we are unable to settle the complexity of this case, we show that the problem becomes equivalent to an interesting variant of the {\sc 2-Page Book Embedding} problem, in which the graph is bipartite and the vertex ordering in the book embedding has to respect the vertex partition of the graph. This problem is in our opinion worthy of future research efforts. For now, we use this equivalence to establish a polynomial-time algorithm for the case in which the link-edges are assigned to the pages of the book embedding.
\item In Section~\ref{se:clique-hierarchy}, we study a Sugiyama-style problem where the cliques are arranged on levels according to a hierarchy. In this practical setting we show that {\sc Clique Planarity} is solvable in polynomial time. This is achieved via a reduction to the {\sc $T$-level planarity} problem~\cite{tibp-addfr-15}.
\end{compactitem}

Conclusions and open problems are presented in Section~\ref{se:conclusions}.


\section{Intersection-link model} \label{se:preliminaries}

Let $G$ be a graph and $S$ be a set of cliques inducing a partition of the vertex set of $G$. In an \emph{intersection-link representation} of $(G,S)$: 

\begin{itemize}
\item each vertex $u$ is a geometric object $R(u)$, which is a translate of an axis-aligned rectangle $\cal R$; 
\item two rectangles $R(u)$ and $R(v)$ intersect if and only if edge $(u,v)$ is an intersection-edge, that is, if and only if $(u,v)$ belongs to a clique in $S$; and 
\item if $(u,v)$ is a link-edge, then it is represented by a curve connecting the boundaries of $R(u)$ and $R(v)$. 
\end{itemize}

To avoid degenerate intersections we assume that no two rectangles have their sides on the same horizontal or vertical line. The {\sc Clique Planarity} problem asks whether an intersection-link representation of a pair $(G,S)$ exists such that: 

\begin{enumerate}
\item no two curves intersect; and 
\item no curve intersects the interior of a rectangle. 
\end{enumerate}

Such a representation is called {\em clique-planar}. A pair $(G,S)$ is {\em clique-planar} if it admits a clique-planar representation.


We now present two simple, yet important, combinatorial properties of intersection-link representations. Let $\Gamma$ be an intersection-link representation of $(K_n,\{K_n\})$ and let $B$ be the outer boundary of $\Gamma$. We have the following.

\begin{lemma} \label{le:two-squares}
Traversing $B$ clockwise, the sequence of encountered rectangles is not of the form $\dots, R(u), \dots, R(v),\\ \dots, R(u), \dots, R(v)$, for any $u,v\in G$.
\end{lemma}

\begin{proof}
The proof is based on the fact that the outer boundary of the union of $R(u)$ and $R(v)$ consists of two maximal portions, one belonging to $R(u)$ and one to $R(v)$.
\end{proof}

\begin{lemma} \label{le:squares-arrangment}
Traversing $B$ clockwise, the sequence of encountered rectangles is a subsequence of $R(u_1),R(u_2),$ $\dots,R(u_n),R(u_{n-1}),\dots,R(u_2)$, for some permutation $u_1,\dots,u_n$ of the vertices of $K_n$.
\end{lemma}

\begin{proof}
We prove a sequence of claims.

(Claim A): Every maximal portion of $B$ belonging to a single rectangle $R(u)$ contains (at least) one corner of $R(u)$. Namely, if part of a side of $R(u)$ belongs to $B$, while its corners do not, then two distinct rectangles $R(v)$ and $R(z)$ enclose those corners. However, this implies that $R(v)$ and $R(z)$ do not intersect, a contradiction to the fact that $\Gamma$ is a representation of $(K_n,\{K_n\})$.

(Claim B): If two adjacent corners of the same rectangle $R(u)$ both belong to $B$, then the entire side of $R(u)$ between them belongs to $B$. Namely, if a rectangle $R(v)\neq R(u)$ intersects a side of $R(u)$, then at least one of the two corners of that side lies in the interior of $R(v)$, given that $R(u)$ and $R(v)$ are translates of the same rectangle; hence that corner does not belong to $B$.

(Claim C): Any rectangle $R(u)$ does not define three distinct maximal portions of $B$. Suppose the contrary, for a contradiction. By Claim A, each maximal portion of $B$ belonging to $R(u)$ contains a corner of $R(u)$. This implies the existence of two adjacent corners belonging to two distinct maximal portions of $B$. However, by Claim B the side of $R(u)$ between those corners belongs to $B$, hence those corners belong to the same maximal portion of $B$, a contradiction.

Claim~C and Lemma~\ref{le:two-squares} imply the statement of the lemma.
\end{proof}

The following lemma allows us to focus, without loss of generality, on special clique-planar representations, which we call {\em canonical}.

\begin{lemma} \label{le:canonical-clique-planar}
Let $(G,S)$ admit a clique-planar representation $\Gamma$. There exists a clique-planar representation $\Gamma'$ of $(G,S)$ such that: 

\begin{itemize}
\item each vertex is represented by an axis-aligned unit square; and 
\item for each clique $s \in S$, all the squares representing vertices in $s$ have their upper-left corner along a common line with slope $1$.
\end{itemize}
\end{lemma}

\begin{proof}
Initialize $\Gamma'=\Gamma$. Rescale $\Gamma'$ in such a way that the unit distance is very small with respect to the size of the rectangles representing vertices in $\Gamma$.

For each clique $s \in S$, consider a closed polyline $P_s$ ``very close'' to the representation of $s$, so that it contains all and only the rectangles representing the vertices of $s$ and it crosses at most once each curve representing a link-edge of $G$. Traverse $P_s$ clockwise. By Lemma~\ref{le:squares-arrangment} and by the clique planarity of $\Gamma$, the circular sequence of encountered curves representing link-edges and crossing $P_s$ contains edges incident to  a subsequence of $R(u_1),R(u_2),\dots,R(u_{|s|}),R(u_{|s|-1}),\dots,R(u_2)$, for some permutation $u_1,\dots,u_{|s|}$ of the vertices of $s$. Remove the interior of $P_s$. Put in the interior of $P_s$ unit squares $Q(u_1),Q(u_2),\dots, Q(u_{|s|})$ representing $u_1,u_2,\dots,u_{|s|}$ as required by the lemma and such that they all share a common point of the plane. Reroute the curves representing link-edges from the border of $P_s$ to the suitable ending squares. This can be done without introducing any crossings, because the circular sequence of the squares encountered when traversing the boundary of the square arrangement clockwise is $Q(u_1),Q(u_2),\dots,Q(u_{|s|}),Q(u_{|s|-1}),\dots,Q(u_2)$.
\end{proof}



\section{Hardness Results on Clique Planarity} \label{se:hardness}

In this section we prove that the {\sc Clique Planarity} problem is not solvable in polynomial time, unless ${\cal P}$=\NP. In fact, we have the following.

\begin{theorem}\label{th:np-one-clique}
It is \NPC to decide whether a pair $(G,S)$ is clique-planar, even if $S$ contains just one clique with more than one vertex.
\end{theorem}

We prove Theorem~\ref{th:np-one-clique} by showing a polynomial-time reduction from a constrained clustered planarity problem, which we prove to be \NPC, to the {\sc Clique Planarity} problem.

A {\em clustered graph} $(G,T)$ is a pair such that $G$ is a graph and $T$ is a rooted tree whose leaves are the vertices of $G$; the internal nodes of $T$ distinct from the root correspond to subsets of vertices of $G$, called {\em clusters}.  A clustered graph is {\em flat} if every cluster is a child of the root.  The {\sc clustered planarity} problem asks whether a given clustered graph $(G,T)$ admits a {\em c-planar drawing}, i.e., a planar drawing of $G$, together with a representation of each cluster $\mu$ in $T$ as a simple region $R_\mu$ of the plane such that: (i) every region $R_\mu$ contains all and only the vertices in $\mu$; (ii) every two regions $R_\mu$ and $R_\nu$ are either disjoint or one contains the other; and (iii) every edge intersects the boundary of each region $R_\mu$ at most once.

Polynomial-time algorithms for testing the existence of a c-planar drawing of a clustered graph are known only in special cases, most notably, if it is {\em c-connected}, i.e., each cluster induces a connected subgraph~\cite{cdfpp-cccg-06,fce-pcg-95}. It has long been known~\cite{fce-pcg-95} that a clustered graph $(G,T)$ is c-planar if and only if a set of edges can be added to $G$ so that the resulting graph is c-planar and c-connected. Any such set of edges is called {\em saturator}, and the subset of a saturator composed of those edges between vertices of the same cluster $\mu$ defines a {\em saturator for $\mu$}. A saturator is {\em linear} if the saturator of each cluster is a path.

The {\sc \cpllong} (\cpl) problem takes as input a flat clustered
graph $(G,T)$ such that each cluster in $T$ induces an independent set
of vertices, and asks whether $(G,T)$ admits a linear saturator.

\newcommand{\lemmaCPLcriterion}{
Let $(G,T)$ be an instance of \cpl with $G=(V,E)$ and let $E^\star \subseteq \binom{V}{2} \setminus E$ be such that in $G^\star = (V,E \cup E^\star)$ every cluster induces a path.  Then $E^\star$ is a linear saturator for $(G,T)$ if and only if $G^\star$ is planar.}

\begin{lemma}
\label{lem:cpl-criterion}
\lemmaCPLcriterion
\end{lemma}

\begin{proof}
  Clearly, if $E^\star$ is a linear saturator, then $(G^\star,T)$ is
  c-planar and thus $G^*$ is planar.  Conversely, assume that $G^\star$ is
  planar and let $\Gamma^\star$ be a planar drawing.  Since the vertices
  of each cluster are isolated in $G$, the region $R_\mu$ for each
  cluster $\mu$ can be represented by a sufficiently narrow region
  around the corresponding path in $G^\star$ yielding a c-planar
  drawing of $G^\star$.  It follows that $E^\star$ is a linear
  saturator.
\end{proof}

The following lemma connects the problem {\sc Clique Planarity} with the problem {\sc \cpllong}.

\begin{lemma} \label{le:reduction-cpl}
Given an instance $(G,T)$ of the \cpl problem, an equivalent instance $(G',S)$ of the {\sc Clique Planarity} problem can be constructed in linear time.
\end{lemma}

\begin{proof}
Instance $(G',S)$ is defined as follows. Initialize $G'=G$. For each cluster $\mu \in T$, add edges to $G'$ such that $\mu$ forms a clique and add this clique to $S$. Clearly, instance $(G',S)$ can be constructed in linear time. We prove that $(G,T)$ admits a linear saturator if and only if $(G',S)$ is clique-planar.

Suppose that $(G,T)$ admits a linear saturator. This implies that there exists a c-planar drawing $\Gamma^\star$ of $(G^\star,T)$, where $G^\star$ is obtained by adding the saturator to $G$. We construct a clique-planar representation $\Gamma$ of $(G',S)$ starting from $\Gamma^\star$ as follows.

Consider a cluster $\mu$ of $T$ represented by region $R_\mu$, let $B_{\mu}$ be the boundary of $R_{\mu}$, and let $u_1,\dots, u_k$ be the vertices of $\mu$ ordered as they appear along the saturator for $\mu$. For each edge $(u,v)$ of $G^\star$ crossing $B_{\mu}$, subdivide $(u,v)$ with a dummy vertex at this crossing point.  Note that the order of the vertices of $\mu$, corresponding to the order in which their incident edges cross $B_\mu$, is a subset of $u_1,\dots,u_{k-1},u_k,u_{k-1}, \dots, u_2$.

Remove from $\Gamma^\star$ all the vertices and (part of the) edges
contained in the interior of $R_\mu$.  Represent $u_1,\dots, u_k$ by pairwise-intersecting rectangles $R(u_1), \dots, R(u_k)$ that are translates of each other and whose upper-left corners touch a common line in this order. Scale $\Gamma^\star$ such that the arrangement can be placed in the interior of $R_\mu$.  Then connect the subdivision vertices on $B_\mu$ with the suitable rectangles.  This is possible without introducing crossings since the  order of the subdivision vertices on $B_\mu$ defines an order of their end-vertices in $\mu$ which is a subsequence of $u_1,\dots,u_{k-1},u_k,u_{k-1}, \dots, u_2$, while the circular order in which the rectangles occur along the boundary of their arrangement is $R(u_1), \dots,R(u_k),R(u_{k-1}), \dots, R(u_2)$. By treating every other cluster of $T$ analogously, we get a clique-planar representation of $(G',S)$.


Conversely, suppose that $(G',S)$ has a clique-planar representation $\Gamma$, which we assume to be canonical by Lemma~\ref{le:canonical-clique-planar}. We define a set $E^\star$ as follows. For each clique $s\in S$, let $R(u_1),\dots,R(u_k)$ be the order in which the rectangles corresponding to $s$ touch the line with slope $1$ through their upper-left corners in $\Gamma$; add to $E^\star$ all the edges $(u_i,u_{i+1})$, for $i=1,\dots,k-1$. We claim that $E^\star$ is a linear saturator for $(G,T)$. Indeed, by Lemma~\ref{lem:cpl-criterion}, it suffices to show that $G+E^\star$ admits a planar drawing.

\begin{figure}[htb]
  \centering
  \subfigure[]{\includegraphics[page=1]{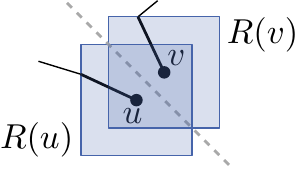}\label{fig:clique-cplanar-a}} \hspace{3mm}
  \subfigure[]{\includegraphics[page=2]{clique-cplanar}\label{fig:clique-cplanar-b}} \hspace{3mm}
  \caption{Construction of a linear saturator from a clique-planar representation.}
  \label{fig:clique-cplanar-reduction}
\end{figure}

Initialize $\Gamma^\star=\Gamma$.  We place each vertex $v$ at the center of square $R(v)$ and remove $R(v)$ from $\Gamma^\star$.  We extend each
edge $(u,v)$ with two straight-line segments from the boundaries of $R(u)$ and $R(v)$ to $u$ and $v$, respectively.  This does not produce crossings; in fact, only the segments of two vertices $u$ and $v$ such that $R(u)$ and $R(v)$ intersect might cross.  However, such segments are separated by the line through the intersection points of the
boundaries of $R(u)$ and $R(v)$; see Fig.~\ref{fig:clique-cplanar-a}. We now draw the edges in $E^\star$ as straight-line segments.
As before, this may not introduce a crossing with any other segment or edge. In fact consider an edge $(u,v)$ in $E^\star$ and any segment $e_w$ incident to a vertex $w\neq u,v$ in the same clique. Assume $u,v,w$ are in this order along the line with slope~1 through them.  Then $(u,v)$ is separated from $e_w$ by the line through the two intersection points of the boundaries of $R(v)$ and
$R(w)$; see Fig.~\ref{fig:clique-cplanar-b}.  This concludes the proof.
\end{proof}


Next, we prove that the \cpl problem is \NPC.

\begin{theorem} \label{th:npc-cp-oneclique}
The \cpl problem is \NPC, even if the underlying graph is a subdivision of a triangulated planar graph and there is just one cluster containing more than one vertex.
\end{theorem}

\begin{proof}
The problem clearly lies in \NP. We give a polynomial-time reduction from the {\sc Hamiltonian Path} problem in biconnected planar graphs~\cite{i-isgci}.

\begin{figure}[htb]
  \centering
  \subfigure[]{\includegraphics[page=2,width=0.25\textwidth]{path-fixed-embedding}\label{fig:linear-saturation-fixed-embedding-c}} \hspace{8mm}
	\subfigure[]{\includegraphics[page=1,width=0.25\textwidth]{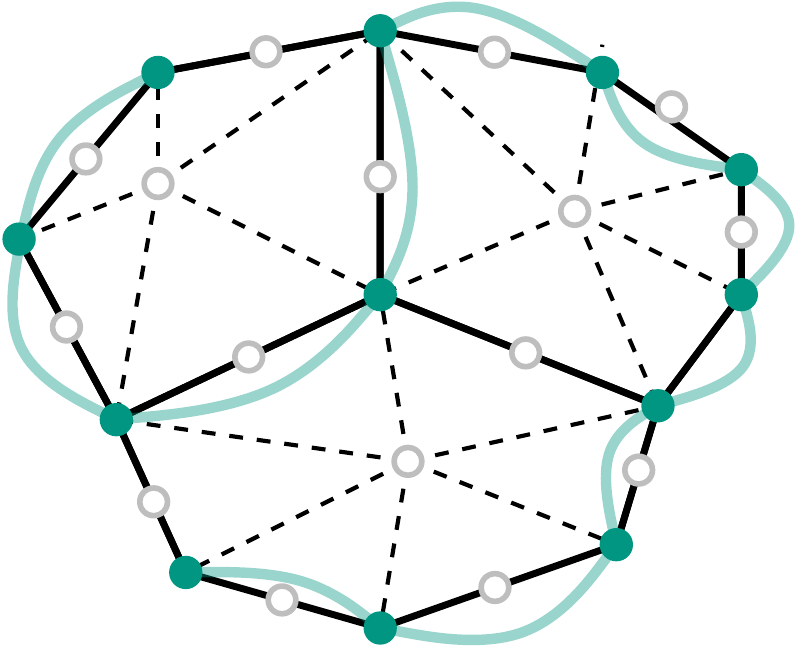}\label{fig:linear-saturation-fixed-embedding-d}}
  \caption{(a) A biconnected planar graph $G$ with a Hamiltonian path $P$. (b) The clustered graph $(G',T)$ obtained from $G$ and the linear saturator for $(G',T)$ corresponding to $P$.}
  \label{fig:clique-cplanar}
\end{figure}

Given a biconnected planar graph $G$ we construct an instance $(G',T)$ of \cpl that admits a linear saturator if and only if $G$ has a Hamiltonian path. Initialize $G'=G$. Let $\cal E$ be a planar embedding of $G'$, as in Fig.~\ref{fig:linear-saturation-fixed-embedding-c}. For each face $f$, add a vertex $v_f$ inside $f$ and connect it to all the vertices incident to $f$. Since $G$ is biconnected, each face $f$ of $\cal E$ is bounded by a simple cycle, hence $G'$ is a triangulated planar graph. Subdivide with a dummy vertex each edge of $G'$ that is not incident to a vertex $v_f$, for any face $f$ of $\cal E$, as in Fig.~\ref{fig:linear-saturation-fixed-embedding-d}.  Finally, add a cluster $\mu$ to $T$  containing all the vertices of $G$ and, for each of the remaining vertices, add to $T$ a cluster containing only that vertex.

Suppose that $G$ admits a Hamiltonian path $P=v_1,\dots,v_n$ and let
$E^\star = \{(v_i,v_{i+1}) \mid 1 \le i \le n-1\}$. Since $P$ is
Hamiltonian, $E^\star$ is a path connecting $\mu$.  Let $G^\star = G' + E^\star$.  Since every cluster different from $\mu$ contains only one vertex, all clusters of
$(G^\star,T)$ induce paths.  A planar drawing of $G^\star$ can be
obtained from a planar drawing $\Gamma$ of $G'$ as follows.  Note
that, for each edge $(v_i,v_{i+1}) \in E^\star$, vertices $v_i$ and
$v_{i+1}$ share two faces in $\Gamma$ since the dummy vertex added to
subdivide edge $(v_i,v_{i+1})$ has degree $2$. Hence, each saturator
edge $(v_i,v_{i+1})$ can be routed inside one of these faces
arbitrarily close to the length-$2$ path between $v_i$ and
$v_{i+1}$ neither crossing an edge of $G'$ nor another saturator edge.  Thus $G^\star$ is planar, and by Lemma~\ref{lem:cpl-criterion} $E^\star$ is a linear saturator for $(G',T)$.

Conversely, suppose $(G',T)$ admits a linear saturator $E^\star$. We claim that $E^\star$ is a Hamiltonian path of $G$. By construction, the vertices of $\mu$ are exactly the vertices of $G$; also, each edge of $E^\star$ corresponds to an edge of $G$, due to the fact that two vertices of $\mu$ are incident to a common face if and only if they are adjacent in $G$. Hence, the path of $G$ corresponding to $E^\star$ is Hamiltonian.
This concludes the proof.
\end{proof}

\section{Clique-Planarity with Given Vertex Representations} \label{se:fixed-representations}

In this section we show how to test {\sc Clique Planarity} in linear time for instances $(G,S)$ with given vertex representations. That is, a clique-planar representation $\Gamma'$ of $(G',S)$ is given, where $G'$ is obtained from $G$ by removing its link-edges, and the goal is to test whether the link-edges of $(G,S)$ can be drawn in $\Gamma'$ to obtain a clique-planar representation $\Gamma$ of $(G,S)$.

We start with a linear-time preprocessing in which we verify that every vertex of $G$ incident to a link-edge is represented in $\Gamma'$ by a rectangle incident to the outer boundary of the clique it belongs to. If the test fails, the instance is negative. Otherwise, we proceed as follows.

\begin{figure}[htb]
  \centering
  \subfigure[]{\includegraphics[page=1,height=.1\textwidth]{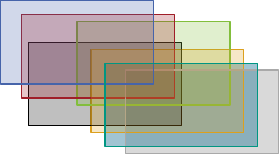}\label{fig:pep-a}}\hspace{1cm}
  \subfigure[]{\includegraphics[page=2,height=.1\textwidth]{pep}\label{fig:pep-b}}\hspace{1cm}
    \subfigure[]{\includegraphics[page=3,height=.1\textwidth]{pep}\label{fig:pep-c}}
  \caption{(a) An intersection-link representation $\Gamma$ of $(K_7,\{s=K_7\})$.
  (b) A simple cycle with a vertex for each maximal portion of the boundary of $\Gamma$ belonging to a single rectangle.
  (c) Planar drawing ${\cal H}'_s$ of graph $H'_s$ corresponding to $\Gamma$.}
  \label{fig:pep}
\end{figure}

We show a reduction to the {\sc Partial Embedding Planarity} problem~\cite{adfjkpr-tppeg-j14}, which asks whether a planar drawing of a graph $H$ exists extending a given drawing ${\cal H}'$ of a subgraph $H'$ of $H$. 

First, we define a connected component $H'_s$ of $H'$ corresponding to a clique $s\in S$ and its drawing ${\cal H}'_s$. We remark that $H'_s$ is a {\em cactus graph}, that is a connected graph that admits a planar embedding in which all the edges are incident to the outer face. Denote by $B$ the boundary of the representation of $s$ in $\Gamma'$ (see Fig.~\ref{fig:pep-a}). If $s$ has one or two vertices, then $H'_s$ is a vertex or an edge, respectively (and ${\cal H}'_s$ is any drawing of $H'_s$). Otherwise,  initialize $H'_s$ to a simple cycle containing a vertex for each maximal portion of $B$ belonging to a single rectangle (see Fig.~\ref{fig:pep-b}). Let ${\cal H}'_s$ be any planar drawing of $H'_s$ with a suitable orientation. Each rectangle in $\Gamma'$ may correspond to two vertices of $H'_s$, but no more than two by Lemma~\ref{le:squares-arrangment}. Insert an edge in $H'_s$ between every two vertices representing the same rectangle and draw it in the interior of ${\cal H}'_s$. By Lemma~\ref{le:two-squares}, these edges do not alter the planarity of ${\cal H}'_s$. Contract the inserted edges in $H'_s$ and ${\cal H}'_s$ (see Fig.~\ref{fig:pep-c}). This completes the construction of $H'_s$, together with its planar drawing ${\cal H}'_s$.  

Graph $H'$ is the union of graphs $H'_s$, over all the cliques $s\in S$; the drawings ${\cal H}'_s$ of $H'_s$ are in the outer face of each other in ${\cal H}'$. Note that, because of the preprocessing, the endvertices of each link-edge of $G$ are vertices of $H'$; then we define $H$ as the graph obtained from $H'$ by adding, for each link-edge $(u,v)$ of $G$, an edge between the vertices of $H'$ corresponding to $u$ and $v$. We have the following:

\begin{lemma} \label{le:pep-correspondence}
There exists a planar drawing of $H$ extending ${\cal H}'$ if and only if there exists a clique-planar representation of $(G,S)$ coinciding with $\Gamma'$ when restricted to $(G',S)$.
\end{lemma}

\begin{proof}
Let $\cal H$ be a planar drawing of $H$ extending ${\cal H}'$. We construct a clique-planar representation $\Gamma$ of $(G,S)$ as follows. Initialize $\Gamma={\cal H}$. For each clique $s \in S$, consider a closed polyline $P_s$ close to ${\cal H}'_s$ so that it contains all and only the vertices and edges of $H'_s$ in its interior and it crosses at most once every other edge of $H$. Scale $\Gamma$ so that, for every clique $s \in S$, a rectangle which is the bounding box of the representation of $s$ in $\Gamma'$ fits in the interior of $P_s$. Remove the interior of $P_s$ and put in its place a copy of the representation of $s$ in $\Gamma'$. Reroute the curves representing link-edges from the border of $P_s$ to the suitable ending rectangles. This can be done without introducing any crossings, because the vertices of $H'_s$ appear along the walk delimiting the outer face of ${\cal H}'_s$ in the same order as the corresponding rectangles appear along the boundary $B$ of the representation of $s$ in $\Gamma'$, by construction. Finally, a homeomorphism of the plane can be exploited to translate the representation of each clique to the position it has in $\Gamma'$, while maintaining the clique planarity of the representation.

Let $\Gamma$ be a clique-planar representation of $(G,S)$. We construct a planar drawing $\cal H$ of $H$ extending ${\cal H}'$ as follows. Initialize ${\cal H}=\Gamma$. For each clique $s \in S$, consider a closed polyline $P_s$  close to the representation of $s$ in $\cal H$ so that it contains all and only the rectangles representing vertices of $s$ and it crosses at most once each curve representing a link-edge of $G$. Remove the interior of $P_s$ and put in its place a scaled copy of ${\cal H}'_s$. Reroute the curves representing link-edges from the border of $P_s$ to the suitable endvertices. As in the previous direction, this can be done without introducing any crossings. Finally, a homeomorphism of the plane can be exploited to transform ${\cal H}$ into a planar drawing that coincides with ${\cal H}'$ when restricted to $H$. 

This concludes the proof of the lemma.
\end{proof}


We get the following main theorem of this section.

\begin{theorem} \label{th:pep-correspondence}
{\sc Clique Planarity} can be decided in linear time for a pair $(G,S)$ if the rectangle representing each vertex of $G$ is given as part of the input.
\end{theorem}

\begin{proof}
First, we check whether, for each $s\in S$, all the rectangles representing vertices in $s$ are pairwise intersecting. This can be done in $O(|s|)$ time by computing the maximum $x$- and $y$-coordinates $x_M$ and $y_M$ among all bottom-left corners, the minimum $x$- and $y$-coordinates $x_m$ and $y_m$ among all top-right corners, and by checking whether $x_M$$<$$x_m$ and $y_M$$<$$y_m$. The described reduction to {\sc Partial Embedding Planarity} can be performed in linear time by traversing the boundary $B$ of each clique $s\in S$; namely, as a consequence of Lemma~\ref{le:squares-arrangment}, $B$ has linear complexity. Contracting an edge requires merging the adjacency lists of its endvertices; this can be done in constant time since these vertices have constant degree, again by Lemma~\ref{le:squares-arrangment}. Further, the {\sc Partial Embedding Planarity} problem can be solved in linear time~\cite{adfjkpr-tppeg-j14}.
\end{proof}

\section{Testing Clique Planarity for Graphs composed of Two Cliques} \label{se:2-cliques}

In this section we study the {\sc Clique Planarity} problem for pairs $(G,S)$ such that $|S|=2$. Observe that, if $|S|=1$, then the {\sc Clique Planarity} problem is trivial, since in this case $G$ is a clique with no link-edge, hence a clique-planar representation of $(G,S)$ can be easily constructed. The case in which $|S|=2$ is already surprisingly non-trivial. Indeed, we could not determine the computational complexity of {\sc Clique Planarity} in this case. However, we establish the equivalence between our problem and a book embedding problem whose study might be interesting in its own; by means of this equivalence we show a polynomial-time algorithm for a special version of the {\sc Clique Planarity} problem. This book embedding problem is defined as follows. 

A {\em $2$-page book embedding} is a plane drawing of a graph where the vertices are cyclically arranged along a closed curve $\ell$, called the {\em spine}, and each edge is entirely drawn in one of the two regions of the plane delimited by $\ell$. The {\sc $2$-Page Book Embedding} problem asks whether a $2$-page book embedding exists for a given graph. This problem is \NPC~\cite{w-chcpmpg-82}. 

Now consider a bipartite graph $G(V_1\cup V_2,E)$. A {\em bipartite $2$-page book embedding} of $G$ is a $2$-page book embedding such that all vertices in $\blue{V_1}$ occur consecutively along the spine (and all vertices in $\red{V_2}$ occur consecutively, as well). We call the corresponding embedding problem {\sc Bipartite $2$-Page Book Embedding} ({\sc b2pbe}).

Finally, we define a {\em bipartite $2$-page book embedding with spine crossings} ({\sc b2pbesc}), as a bipartite $2$-page book embedding in which edges are not restricted to lie in one of the two regions delimited by $\ell$, but each of them might cross $\ell$ once. These crossings are only allowed to happen in the two portions of $\ell$ delimited by a vertex of $V_1$ and a vertex of $V_2$. We call the corresponding embedding problem {\sc Bipartite $2$-Page Book Embedding with Spine Crossings} ({\sc b2pbesc}).


We now prove that the {\sc b2pbesc} problem is equivalent to {\sc Clique Planarity} for instances $(G,S)$ such that $|S|=2$. Consider any instance $(G',\{s_1,s_2\})$ of the {\sc Clique Planarity}  problem. An instance \bookinstance{}  of the {\sc b2pbesc} problem can be defined in which $V_1$ is the vertex set of $s_1$ and $V_2$ is the vertex set of $s_2$; also, $E$ consists of all the link-edges of $G'$. Conversely, given an instance \bookinstance{} of {\sc b2pbesc}, an instance $(G',\{s_1,s_2\})$ of {\sc Clique Planarity} can be constructed in which $s_1$ is a clique on $V_1$ and $s_2$ is a clique on $V_2$; the set of link-edges of $G'$ coincides with $E$. Observe that, since link-edges only connect vertices of different cliques and since edges of $E$ only connect a vertex of $V_1$ to one of $V_2$, each mapping generates a valid instance for the other problem. Also, these mappings define a bijection, hence the following lemma establishes the equivalence between the two problems.

\begin{lemma} \label{le:bipartite-equivalence}
$(G',\{s_1,s_2\})$ is clique-planar if and only if \bookinstance{} admits a {\sc b2pbesc}.
\end{lemma}

\begin{proof}
Suppose that there exists a {\sc b2pbesc} $\cal B$ of \bookinstance{}. We construct a clique-planar representation $\Gamma$ of $(G',S)$ as follows. Initialize $\Gamma={\cal B}$. Relabel the vertices in $V_1$ (resp. in $V_2$) as $u_1,\dots,u_k$ (resp. $v_1,\dots,v_h$) according to the order in which they appear along $\ell$. Draw a closed curve $\lambda_1$ ($\lambda_2$) enclosing a portion of the spine $\ell$ containing all and only vertices $u_1,\dots,u_k$ (resp. $v_1,\dots,v_h$). Scale $\Gamma$ so that $\lambda_1$ and $\lambda_2$ are large enough to contain a square of size $(1+\epsilon) \times (1+\epsilon)$ in their interiors, with $\epsilon>0$. Remove the interior of $\lambda_1$ and $\lambda_2$. Draw pairwise-intersecting unit squares $Q(u_1),\dots,Q(u_k)$ (resp. $Q(v_1),\dots, Q(v_h)$) all in the interior of $\lambda_1$ (resp. of $\lambda_2$) with their upper-left corners in this order along a common line $l_1$ (resp. $l_2$).  Reroute the curves representing portions of link-edges from the border of $\lambda_1$ and $\lambda_2$ to the suitable ending squares inside them. This can be done without introducing any crossings, because the vertices of $V_1$ ($V_2$) appear along $\ell$ in the same order as the corresponding squares touch $l_1$ (resp. $l_2$); also, the portion of a link-edge connecting a point on $\lambda_1$ with a point on $\lambda_2$ is contained in the original drawing of the edge of $G$, hence no two such portions cross each other. Thus, $\Gamma$ is a clique-planar representation of $(G',S)$.

\begin{figure}[htb]
\centering
        \includegraphics[height=.3\textwidth,page=1]{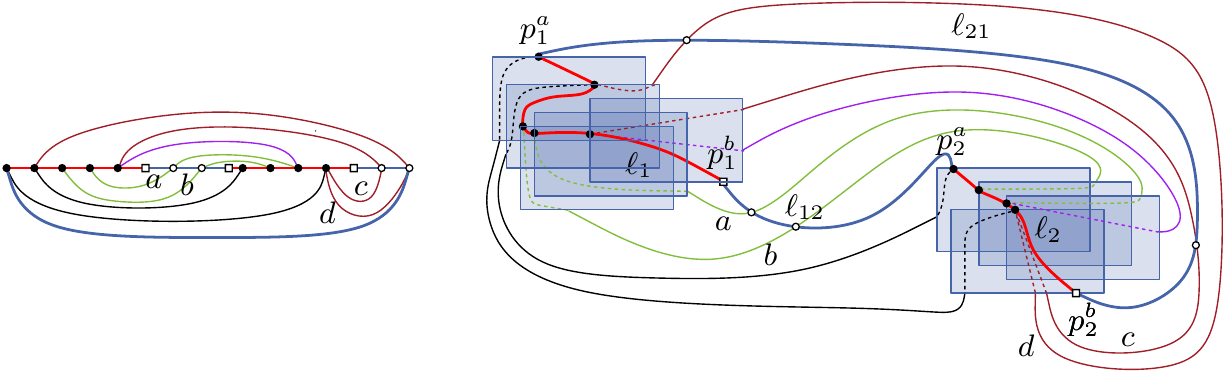} \label{fig:2cliquetoSEFE-a}
        \caption{Constructing a {\sc b2pbesc} from a clique-planar representation.}
\label{fig:book-to-clique}
    \end{figure}

Suppose that there exists a clique-planar representation $\Gamma$ of $(G',\{s_1,s_2\})$. We construct a {\sc b2pbesc} $\cal B$ of \bookinstance{} as follows. Refer to Fig.~\ref{fig:book-to-clique}. Initialize ${\cal B}=\Gamma$. By Lemma~\ref{le:squares-arrangment}, the order in which the rectangles representing the vertices $u_1, \dots, u_k$ in $s_1$ are encountered when traversing the boundary $B_1$ of their arrangement clockwise, is a subsequence of $R(u_1), \dots, R(u_k), R(u_{k-1}), \dots, R(u_2)$. Hence, for any two points $p_1^a$ on $R(u_1) \cap B_1$ and $p_1^b$ on $R(u_k) \cap B_1$, there exists a curve $\ell_1$ between $p_1^a$ and $p_1^b$ entering $R(u_1), \dots, R(u_k)$ in this order. Place vertex $u_i$ of $V_1$ at the point where $\ell_1$ enters $R(u_i)$. Define $\ell_2$, $p_2^a$ and $p_2^b$, and draw the vertices of $V_2$ analogously. Further, add to $\cal B$ two curves $\ell_{12}$ and $\ell_{21}$, not intersecting each other, not intersecting the same link-edge, each intersecting a link edge at most once, and connecting $p_1^a$ to $p_2^b$, and $p_2^a$ to $p_1^b$, respectively. Reroute the curves representing portions of the edges in $E$ from $B_1$ and $B_2$ to the suitable ending vertices inside them. This can be done without introducing any crossings, because the vertices of $V_1$ ($V_2$) appear along $\ell_1$ (along $\ell_2$) in the same order as the corresponding rectangles appear along $B_1$ (along $B_2$) in $\Gamma$, by construction. Finally, consider the curve $\ell$ composed of $\ell_1$, $\ell_{12}$, $\ell_2$, and $\ell_{21}$. We have that all the vertices of $V_1$ (of $V_2$) appear consecutively along $\ell$, since they all lie on $\ell_1$ (on $\ell_2$); also, each edge $e \in E$ crosses $\ell$ at most once, either on $\ell_{12}$ or on $\ell_{21}$. Hence, $\cal B$ is a {\sc b2pbesc} of \bookinstance{}. This concludes the proof of the lemma.
\end{proof}


We now consider a variant of the {\sc Clique Planarity} problem for two cliques in which each clique is associated with a $2$-partition of the link-edges incident to it, and the goal is to construct a clique-planar representation in which the link-edges in different sets of the partition exit the clique on ``different sides''. This constraint finds a correspondence with the variant of the (non-bipartite) $2$-page book embedding problem, called {\sc Partitioned $2$-page book embedding} problem, in which vertices are allowed to be arbitrarily permuted along the spine, while the edges are pre-assigned to the pages of the book~\cite{hn-sattpbep-14}. 

\begin{figure}[htb]
\centering
        \includegraphics[height=.2\textwidth]{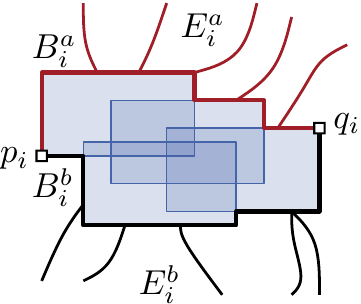} 
        \caption{An intersection-link representation $\Gamma_i$ of $s_i$. The top side of $\Gamma_i$ and the link-edges in $E^a_i$ are red, while the bottom side of $\Gamma_i$ and the link-edges in $E^b_i$ are black.}
\label{fig:clique-notation}
\end{figure}

More formally, let $(G,S=\{s_1,s_2\})$ be an instance of {\sc Clique Planarity} and let $\{E^a_i, E^b_i\}$ be a partition of the link-edges incident to $s_i$, with $i \in \{1,2\}$. Consider an intersection-link representation $\Gamma_i$ of $s_i$ with outer boundary $B_i$, let $p_i$ be the bottom-left corner of the leftmost rectangle in $\Gamma_i$, and let $q_i$ be the upper-right corner of the rightmost rectangle in $\Gamma_i$; see Fig.~\ref{fig:clique-notation}. Let $B^a_i$ be the part of $B_i$ from $p_i$ to $q_i$ in clockwise direction (this is the {\em top side} of $\Gamma_i$) and let $B^b_i$ be the part of $B_i$ from $q_i$ to $p_i$ in clockwise direction (this is the {\em bottom side} of $\Gamma_i$). We aim to construct a clique-planar representation of $(G,S)$ in which all the link-edges in $E^a_i$ (resp. in $E^b_i$) intersect the arrangement $\Gamma_i$ of rectangles representing $s_i$ on the top side (resp. on the bottom side) of $\Gamma_i$. We call the problem of determining whether such a representation exists {\sc $2$-Partitioned Clique Planarity}. We prove that {\sc $2$-Partitioned Clique Planarity} can be solved in quadratic time. 

The algorithm is based on a reduction to equivalent special instances of {\sc Simultaneous Embedding with Fixed Edges} (\sefe) that can be decided in quadratic time. Given two graphs \Gr{} and \Gb{} on the same vertex set $V$, the \sefe problem asks to find planar drawings of $G_1$ and $G_2$ that coincide on $V$ and on the common edges of \Gr{} and \Gb{}. We have the following.

\begin{lemma}\label{le:2cliquetoSEFE}
Let $(G,\{s_1,s_2\})$ and $\{E^a_1, E^b_1,E^a_2, E^b_2\}$ be an instance of {\sc $2$-Partitioned Clique Planarity}. An equivalent instance \sefeinstance{} of \sefe such that $\red{G_1}=(V,\red{E_1})$ and $\blue{G_2}=(V,\blue{E_2})$ are $2$-connected and such that the common graph $G_\cap = (V,\red{E_1}\cap\blue{E_2})$ is connected can be constructed in linear time.
\end{lemma}

\begin{proof}
By Lemma~\ref{le:bipartite-equivalence}, we can describe $(G,\{s_1,s_2\})$ by its equivalent instance \bookinstance{} of {\sc b2pbesc}, where $V_1$ is the vertex set of $s_1$, $V_2$ is the vertex set of $s_2$, and $E$ is the set of link-edges of $(G,\{s_1,s_2\})$. The partition $\{E^a_i, E^b_i\}$ of the edges incident to $s_i$ translates to constraints on the side of the spine $\ell$ of the {\sc b2pbesc} each of these edges has to be incident to. Namely, for each vertex $u_i \in V_1$, all the edges in $E^a_1$ (in $E^b_1$) incident to $u_i$ have to exit $u_i$ from the internal (resp.\ external) side of $\ell$; and analogously for the edges of $E^a_2$ and $E^b_2$. This implies that edges in $E^a_1 \cap E^a_2$ entirely lie inside $\ell$, edges in $E^b_1 \cap E^b_2$ entirely lie outside $\ell$, while the other edges have to cross $\ell$.

\begin{figure}[htb]
\centering
        {\includegraphics[width=\textwidth,page=2]{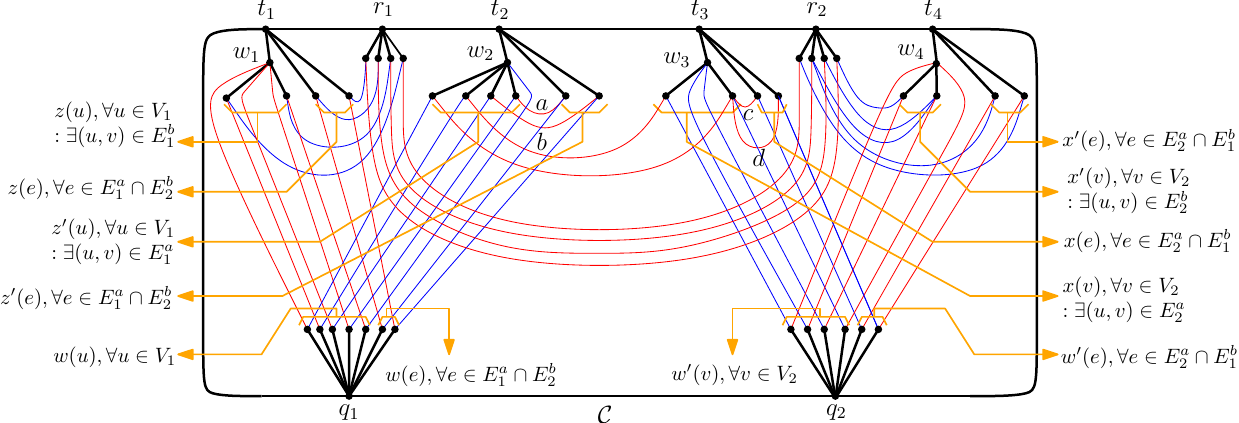} \label{fig:2cliquetoSEFE-b}}
        \caption{The {SEFE} $\langle \red{\Gamma_1},\blue{\Gamma_2}\rangle $ of \sefeinstance{} corresponding to $\Gamma$ from Fig.~\ref{fig:book-to-clique}.}\label{fig:2cliquetoSEFE}
    \end{figure}

We now describe how to construct \sefeinstance{}.  Refer to Fig.~\ref{fig:2cliquetoSEFE}. 

The common graph \Gint contains a cycle $\mathcal{C} = (t_1, r_1, t_2, t_3, r_2, t_4, q_2, q_1)$.  Also, \Gint contains two stars $Q_1$ and $Q_2$ centered at $q_1$ and $q_2$, respectively, where $Q_1$ has a leaf $w(u_i)$ for each vertex $u \in V_1$ and a leaf $w(e)$ for each edge $e \in E^a_1 \cap E^b_2$, and where $Q_2$ has a leaf $w'(v)$ for each vertex $v \in V_2$ and a leaf $w'(e)$ for each edge $e \in E^b_1 \cap E^a_2$. Further, \Gint contains trees $T_i$ rooted at $t_i$, for $i=1,\dots,4$, defined as follows.

Tree $T_1$ contains a leaf $z(e)$ adjacent to $t_1$ for each edge $e \in E^a_1 \cap E^b_2$; also, it contains a vertex $w_1$ adjacent to $t_1$; finally, it contains a leaf $z(u)$, adjacent to $w_1$, for each vertex $u \in V_1$ that is incident to at least one edge in $E^b_1$.

Tree $T_2$ contains a leaf $z'(e)$ adjacent to $t_2$ for each edge $e \in E^a_1 \cap E^b_2$; also, it contains a vertex $w_2$ adjacent to $t_2$; finally, it contains a leaf $z'(u)$, adjacent to $w_2$, for each vertex $u \in V_1$ that is incident to at least one edge in $E^a_1$.

Tree $T_3$ contains a leaf $x(e)$ adjacent to $t_3$ for each edge $e \in E^a_2 \cap E^b_1$; also, it contains a vertex $w_3$ adjacent to $t_3$; finally, it contains a leaf $x(v)$, adjacent to $w_3$, for each vertex $v \in V_2$ that is incident to at least one edge in $E^a_2$.

Tree $T_4$ contains a leaf $x'(e)$ adjacent to $t_4$ for each edge $e \in E^a_2 \cap E^b_1$; also, it contains a vertex $w_4$ adjacent to $t_4$; finally, it contains a leaf $x'(v)$, adjacent to $w_4$, for each vertex $v \in V_2$ that is incident to at least one edge in $E^b_2$.

Finally, \Gint contains two stars $R_1$ and $R_2$ centered at $r_1$ and $r_2$, respectively, with the same number of leaves as $T_1$. Namely, $R_1$ ($R_2$) contains a leaf $r_1(e)$ (a leaf $r_2(e)$) adjacent to $r_1$ (to $r_2$, resp.) for each edge $e \in E^a_1 \cap E^b_2$; also, it contains a leaf $r_1(u)$ (a leaf $r_2(u)$, resp.) for each vertex $u \in V_1$ that is incident to at least one edge in $E^b_1$.

Graph \Gr{} contains \Gint plus the following edges. Consider each edge $e = (u,v)$ with $u \in V_1$ and $v \in V_2$.
If $e \in E^a_1 \cap E^b_2$, graph \Gr{} has an edge $(w(e),z(e))$ and an edge $(z'(u),z'(e))$;
if $e \in E^b_1 \cap E^a_2$, graph \Gr{} has an edge $(w'(e),x'(e))$ and an edge $(x(v),x(e))$;
if $e\in E^a_1 \cap E^a_2$, graph \Gr{} has an edge $(z'(u),x(v))$.
For each vertex $u \in V_1$, if $u$ is incident to at least one edge in $E^b_1$, then \Gr{} contains edge $(w(u),z(u))$ and edge $(r_1(u),r_2(u))$, otherwise it contains edge $(w(u),w_1)$.
For each vertex $v \in V_2$, if $v$ is incident to at least one edge in $E^b_2$, then \Gr{} contains edge $(w'(v),x'(v))$, otherwise it contains edge $(w'(v),w_4)$.

Graph \Gb{} contains \Gint plus the following edges. Consider each edge $e = (u,v)$ with $u \in V_1$ and $v \in V_2$.
If $e \in E^a_1 \cap E^b_2$, graph \Gb{} has an edge $(w(e),z'(e))$, an edge $(z(e),r_1(e))$, and an edge $(r_2(e),x'(v))$.
If $e \in E^a_2 \cap E^b_1$, graph \Gb{} has an edge $(r_2(u),x'(e))$ and an edge $(w'(e),x(e))$.
If $e \in E^b_1 \cap E^b_2$, graph \Gr{} has an edge $(r_2(u),x'(v))$.
For each vertex $u \in V_1$, if $u$ is incident to at least one edge in $E^a_1$, then \Gb{} contains edge $(w(u),z'(u))$, otherwise it contains edge $(w(u),w_2)$. Also, if $u$ is incident to at least one edge in $E^b_1$, then \Gb{} contains edge $(w(u),r_1(u))$.
For each vertex $v \in V_2$, if $v$ is incident to at least one edge in $E^a_2$, then \Gb{} contains edge $(w'(v),x(v))$, otherwise it contains edge $(w'(v),w_3)$.

The order of the leaves of $Q_1$ represents the order in which the vertices of $s_1$ and the crossings between edges of $E^a_2 \cap E^b_1$ appear along $\ell$; analogously, the order of the leaves of $Q_2$ represents the order in which the vertices of $s_2$ and the crossings between edges of $E^a_1 \cap E^b_2$ appear along $\ell$. Hence, concatenating these two orders yields a total order of the vertices and of the crossings along $\ell$. In the following we prove that this correspondence determines a {\sc b2pbesc}, when starting from a \sefe, and vice versa.

Suppose that \sefeinstance{} admits a \sefe. Let $\mathcal{O}_1$ and $\mathcal{O}_2$ be the orders in which the leaves of $Q_1$ and of $Q_2$ appear around $q_1$ and $q_2$, respectively.
We obtain a {\sc b2pbesc} $\mathcal{B}$ by concatenating $\mathcal{O}_1$ and $\mathcal{O}_2$, where each vertex $u_i$ of $V_1$ corresponds to leaf $w(u_i)$ of $Q_1$, each vertex $v_j$ of $V_2$ corresponds to leaf $w'(v_j)$ of $Q_2$, and the crossing between an edge $e$ and the spine $\ell$ of $\mathcal{B}$ corresponds to either leaf $w(e)$ of $Q_1$ or leaf $w'(e)$ of $Q_2$, depending on whether $e \in E^b_1 \cap E^a_2$ or $e \in E^a_1 \cap E^b_2$.

Since the leaves of $Q_1$ corresponding to vertices of $V_1$ are all connected to vertex $w_1$ by paths belonging to the same graph \Gr{}, they appear consecutively around $q_1$, and hence the corresponding vertices of $V_1$ appear consecutively along $\ell$. Analogously, all the vertices of $V_2$ appear consecutively along $\ell$, as the corresponding leaves of $Q_2$ are all connected to $w_4$. These two facts imply that all the edges cross the spine in a point lying between a vertex of $V_1$ and a vertex of $V_2$. Hence, the order of the vertices along the spine $\ell$ of $\cal B$ is consistent with a valid \bTWObesc. We now show that this order also allows drawing the edges without crossings.

The routing of each edge $e = (u_i,v_j)$ is performed as follows. If $e \in E^a_1 \cap E^a_2$, then $e$ is drawn as a curve on the internal side of $\ell$.
If $e \in E^b_1 \cap E^b_2$, then $e$ is drawn as a curve on the external side of $\ell$. If $e \in E^a_1 \cap E^b_2$, then $e$ is drawn as a curve whose portion between $u_i$ and the crossing point $y_e$ is on the internal side of $\ell$, and whose portion between $y_e$ and $v_j$ is on the external side of $\ell$. If $e \in E^b_1 \cap E^a_2$, then $e$ is drawn as a curve whose portion between $u_i$ and $y_e$ is on the external side of $\ell$, and whose portion between $y_e$ and $v_j$ is on the internal side of $\ell$.

First observe that, because of the edges connecting $Q_1$ with $T_1$ and $T_2$, the clockwise order of the leaves of $Q_1$ coincides with the counterclockwise order of the leaves of $T_1$ and $T_2$, when restricting these orders to the leaves corresponding to the same vertices and edges. The same holds for $Q_2$ with respect to $T_3$ and $T_4$. Also, because of the edges connecting $R_1$ with $T_1$ and $R_2$, the clockwise order of the leaves of $R_2$ is the same as the one of the leaves of $T_1$.

We now prove that no two edges in the constructed book embedding cross each other. Observe that each edge either entirely lies in one of the two sides of $\ell$, or it crosses $\ell$ once, hence it is composed of two portions in different sides of $\ell$. Clearly, it suffices to prove that no two edges (or portions) on the same side of $\ell$ cross each other.

Consider the portions of the edges of \bookinstance{} that lie on the same side of $\ell$, say on the internal side of $\ell$: these are edges $e=(u_i,v_j)$ in $E^a_1 \cap E^a_2$, the portions of the edges $e = (u_m,v_h)$ in $E^a_1 \cap E^b_2$ between $u_m$ and $y_e$, and the portions of the edges $e = (u_l,v_p)$ in $E^b_1 \cap E^a_2$ between $y_e$ and $v_p$. Every edge of the first type corresponds to an edge $(z'(u_i),x(v_j))$ of \Gr{}; every portion of an edge of the second type corresponds to an edge $(z'(u_m),z'(e))$ of \Gr{}; and every portion of an edge of the third type corresponds to an edge $(x(e),x(v_p))$ of \Gr{}.
Since all these edges belong to \Gr{}, they do not cross in the given \sefe.
By construction, for each vertex of $G$ incident to at least one of the considered edges, there is a corresponding leaf of either $T_2$ or $T_3$ in \Gint; also, for each of these edges that crosses $\ell$, there is a corresponding leaf of either $T_2$ or $T_3$ in \Gint. Further, these leaves appear, in the left-to-right order of the leaves of $T_2$ and $T_3$, in the same order as they appear along $\ell$, by the previous observation that the order of the leaves of $T_2$ (of $T_3$) coincides with the one of $Q_1$ (of $Q_2$, resp.). This implies that the considered edges of \bookinstance{} do not cross in the internal side of $\ell$.
Analogous arguments can be used to prove that the portions of the edges of \bookinstance{} lying in the external side of $\ell$ do not cross each other. In this case, the edges of the \sefe that have to be considered are $(r_2(u_i),x'(v_j))$, $(r_2(u_i),x'(e))$, and $(r_2(e),x'(v_j))$, which all belong to \Gb{}.

We now prove the opposite direction. Suppose that \bookinstance{} admits a {\sc b2pbesc} $\mathcal{B}$. We construct a \sefe of \sefeinstance{} as follows.

We define a linear ordering $\sigma_1$ of the vertices of $V_1$ and the crossings between $\ell$ and the edges in $E^a_1 \cap E^b_2$; $\sigma_1$ is the clockwise order in which such vertices and crossings appear along $\ell$ starting at any vertex of $V_2$. Observe that, since $\mathcal{B}$ is a {\sc b2pbesc}, all the vertices in $V_1$ appear consecutively in $\sigma_1$. Analogously, we define a linear ordering $\sigma_2$ of the vertices of $V_2$ and the crossings between $\ell$ and the edges in $E^a_2 \cap E^b_1$; $\sigma_2$  is the clockwise order in which such vertices and crossings appear along $\ell$ starting at any vertex of $V_1$. Observe that, since $\mathcal{B}$ is a {\sc b2pbesc}, all the vertices in $V_2$ appear consecutively in $\sigma_2$.

Recall that each leaf of $Q_1$ either corresponds to a vertex in $V_1$ or to an edge in $E^a_1 \cap E^b_2$. Thus, we can define a clockwise linear ordering of the leaves of $Q_1$ around $q_1$ as the corresponding vertices and crossings appear in $\sigma_1$; this linear ordering starts after edge $(q_1,t_1)$. A clockwise linear ordering of the leaves of $Q_2$ around $q_2$ is defined analogously from $\sigma_2$ starting from edge $(q_2,q_1)$. 

The order of the leaves of trees $T_i$, with $i=1,2,3,4$, and of stars $R_1$ and $R_2$ is also decided based on $\sigma_1$ and $\sigma_2$. This allows us to draw all the edges of \Gr{} and \Gb{} that are incident to stars $Q_1$, $Q_2$, $R_1$, and $R_2$ without crossings; in particular, the paths connecting $q_1$ to $w_1$ and $w_2$, and those connecting $q_2$ to $w_3$ and $w_4$ can be drawn without crossings since the vertices of $V_1$ (and the vertices of $V_2$) are consecutive along $\ell$ in $\mathcal{B}$. The fact that the edges connecting the leaves of $T_2$ and $T_3$, and the edges connecting the leaves of $R_2$ and $T_4$ can be drawn without crossings is again due to the fact that, by construction, these edges correspond to portions of edges of $G$ lying on the same side of $\ell$.

Graph \Gint is connected, by construction. The fact that \Gr{} and \Gb{} are $2$-connected can be proved as follows. Graph \Gr{} is composed of the outer cycle $\cal C$, plus a set of $2$-connected components connecting pairs of vertices of $\cal C$. One component connects $q_1$ to $t_1$; one component connects $r_1$ to $r_2$; one connects $q_2$ to $t_4$; and another one connects $t_2$ to $t_3$. In particular, in order for this latter component to actually exist, at least one edge in $E^a_1 \cap E^a_2$ must exist. However, this fact can be assumed without loss of generality, as otherwise two dummy vertices and an edge in $E^a_1 \cap E^a_2$ between them could be added to the instance without altering the possibility of finding a {\sc b2pbesc}. As for \Gb{}, it is also composed of the outer cycle $\cal C$, plus a set of $2$-connected components connecting pairs of vertices of $\cal C$. One component connects $q_1$ to $t_2$; one connects $t_1$ to $r_1$; one connects $r_2$ to $t_4$; and another one connects $q_2$ to $t_3$.

As \sefeinstance{} can be easily constructed in linear time, the lemma follows.
\end{proof}


\begin{theorem}\label{th:partitioned2clique-polynomial}
Problem {\sc $2$-Partitioned Clique Planarity} can be solved in quadratic time for instances $(G,S)$ in which $|S|=2$.
\end{theorem}
\begin{proof}
Apply Lemma~\ref{le:2cliquetoSEFE} to construct in linear time an instance \sefeinstance{} of \sefe that is equivalent to $(G,S)$ such that \Gr{} and \Gb{} are biconnected and their intersection graph \Gint is connected. The statement follows from the fact that instances of \sefe with this property can be solved in quadratic time~\cite{br-spacep-13}.
\end{proof}

\section{Clique Planarity with Given Hierarchy}\label{se:clique-hierarchy}

In this section we study a version of the {\sc Clique Planarity} problem in which the cliques are given together with a hierarchical relationship among them.
Namely, let $(G,S)$ be an instance of {\sc Clique Planarity} and let $\psi: S \rightarrow \{1,\dots,k\}$, with $k \leq |S|$, be an assignment of the cliques in $S$ to $k$ levels such that, for each link-edge $(u,v)$ of $G$ connecting a vertex $u$ of a clique $s'$ to a vertex $v$ of a clique $s''$, we have $\psi(s') \neq \psi(s'')$; an instance is {\em proper} if $\psi(s') = \psi(s'') \pm 1$ for each link-edge.

We aim to construct canonical clique-planar representations of $(G,S)$ such that:

\begin{itemize}
\item (Property~1) for each clique $s \in S$, the top side of the bounding box of the representation of $s$ lies on the line $y=2\psi(s)$, while the bottom side lies above the line $y=2\psi(s)-2$; and 
\item (Property~2) each link-edge $(u,v)$, with $u \in s'$, $v \in s''$, $\psi(s') < \psi(s'')$, is drawn as a $y$-monotone curve from the top side of $R(u)$ to the bottom side of $R(v)$.
\end{itemize}

We call the problem of testing whether such a representation exists {\sc Level Clique Planarity}.

We show how to test level clique planarity in quadratic time for proper instances via a linear-time reduction to equivalent proper instances of {\sc T-level Planarity}~\cite{tibp-addfr-15}.

A \emph{$\mathcal{T}$-level graph} $(V,E,\gamma,\mathcal{T})$ consists of: 
\begin{enumerate}[(i)]
\item a graph $G =(V,E)$;
\item a function $\gamma: V \rightarrow \{1,...,k\}$ such that $\gamma(u) \neq \gamma(v)$ for each $(u,v) \in E$, where the set $V_i = \{v \mid \gamma(v)=i\}$ is the $i$-th \emph{level} of $G$; and 
\item a set $\mathcal{T}=\{T_1,\dots,T_k\}$ of rooted trees such that the leaves of $T_i$ are the vertices in $V_i$.
\end{enumerate}

A \emph{$\mathcal{T}$-level planar drawing} of $(V,E,\gamma,\mathcal{T})$ is a planar drawing of $G$ where the edges are $y$-monotone curves and the vertices in $V_i$ are placed along the line $y = i$, denoted by $\ell_i$, according to an order \emph{compatible} with $T_i$; that is, for each internal node $\mu$ of $T_i$, the leaves of the subtree of $T_i$ rooted at $\mu$ are consecutive along $\ell_i$. A $\mathcal{T}$-level graph is {\em $\mathcal{T}$-level planar} if it admits a $\mathcal{T}$-level planar drawing. The \tlp problem asks to test whether a $\mathcal{T}$-level graph is $\mathcal{T}$-level planar. We have the following.

\begin{lemma}\label{le:hierarchical-2-cyclicTlevel}
Given a proper instance of {\sc Level Clique Planarity}, an equivalent proper instance of {\sc T-level Planarity} can be constructed in linear time.
\end{lemma}

\begin{proof}
Given $(G(V,E),S,\psi)$, an instance $(V,E',\gamma,\mathcal{T})$ of {\sc T-level Planarity} can be constructed as follows. The vertex sets of the graphs coincide and $E'$ coincides with the set of link-edges in $E$. For each vertex $v$ in a clique $s \in S$ we have $\gamma(v)=\psi(s)$. Finally, for $i=1,\dots,k$, where $k$ is the number of levels in $(G,S,\psi)$, tree $T_i \in \mathcal{T}$ has root $r_i$, a child $w_s$ of $r_i$ for each $s\in S$, and the vertices of $s$ as children of $w_s$. 

\begin{figure}[htb]
  \centering
  \subfigure[]{\includegraphics[page=2,height=.16\textwidth]{LevelCliquePlanarity}\label{fig:level-clique-a}}\hspace{1cm}
  \subfigure[]{\includegraphics[page=1,height=.16\textwidth]{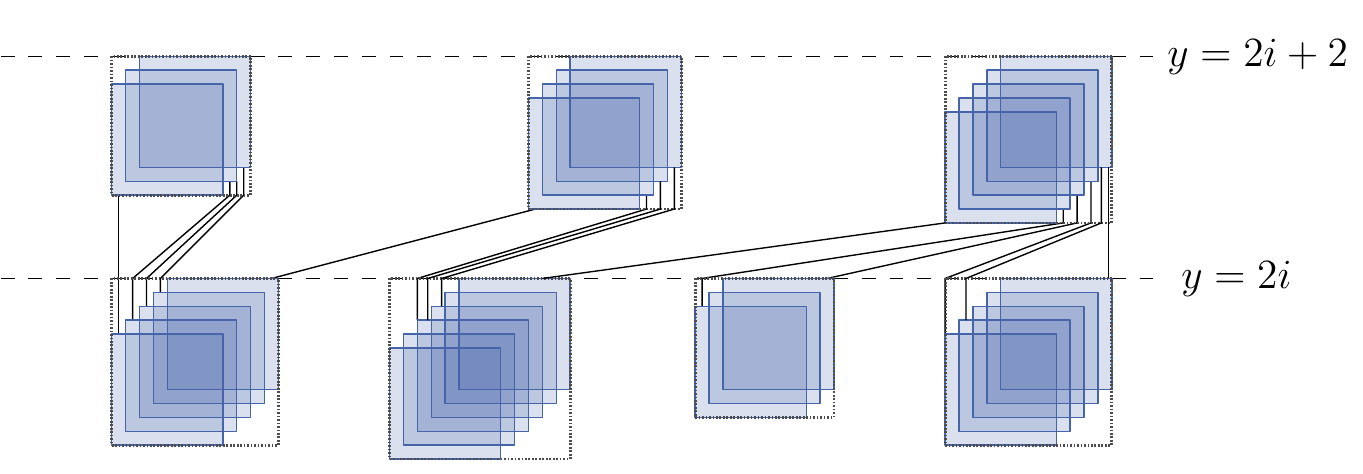}\label{fig:level-clique-b}}
  \caption{Construction of a clique-planar representation of $(G(V,E),S,\psi)$ from a $T$-level planar drawing $\Gamma$ of $(V,E',\gamma,\mathcal{T})$. (a) The part of $\Gamma$ between two levels $i$ and $i+1$. The edges and the internal nodes of trees $T_i$ and $T_{i+1}$ are green, while the vertices in $V$ and the edges in $E'$ are black. (b) The corresponding clique-planar representation between levels $i$ and $i+1$. The bounding box of the representation of each clique is dotted.}
  \label{fig:level-clique}
\end{figure}

Suppose that $(V,E',\gamma,\mathcal{T})$ admits a $T$-level planar drawing $\Gamma$, as in Fig.~\ref{fig:level-clique-a}. We construct a clique-planar representation with the desired properties as follows; refer to Fig.~\ref{fig:level-clique-b}. For each clique $s \in S$ with $\psi(s)=i$, we construct a canonical representation of $s$ in a bounding box of size $(1+\varepsilon)\times(1+\varepsilon)$, with $0<\varepsilon<1$, and plug it between lines  $y=2\psi(s)$ and $y=2\psi(s)-2$ with the top side of the bounding box lying on line $y=2\psi(s)$. Note that the bottom side of the bounding box is above the line $y=2\psi(s)-2$. Cliques on the same level are placed side-by-side, so that they do not touch each other. Finally, for each two consecutive levels $V_i$ and $V_{i+1}$, consider the edges in $E'$ connecting a vertex $u \in V_i$ with a vertex $v \in V_{i+1}$ as they appear in $\Gamma$ from left to right; we draw the corresponding link-edge $(u,v) \in E$ as a polyline, lying to the right of any previously drawn edge between $V_i$ and $V_{i+1}$, composed of three segments: the first is a vertical segment connecting a point on the top side of $R(u)$ with a point $p_u$ on the top side of the bounding box of $s'$, the third is a vertical segment connecting a point on the bottom side of $R(v)$ with a point $p_v$ on the bottom side of the bounding box of $s''$, the middle one is a straight-line segment connecting $p_u$ with $p_v$.

The obtained representation satisfies Properties~1 and~2 by construction.  That the link-edges do not cross each other descends from the following two facts. First, the node $w_s$ in $T_{\psi(s)}$ enforces vertices of the same clique $s$ to be consecutive along $\ell_{\psi(s)}$ in $\Gamma$. Second, for each clique $s \in S$, the squares representing the vertices of $s$ are in the same order as the corresponding vertices of $s$ along $\ell_{\psi(s)}$. This implies that, for any two cliques $s' \in V_i$ and $s'' \in V_{i+1}$, the left-to-right order in which the link-edges between $s'$ and $s''$ intersect the line $y=2\psi(s')$ is the same as the one in which they intersect the line $y=2(\psi(s''))-(1+\varepsilon)$, hence no two such link-edges cross.

Suppose that $(G(V,E),S,\psi)$ admits a clique-planar representation satisfying Properties~1 and~2. We construct a $T$-level planar drawing $\Gamma$ of $(V,E',\gamma,\mathcal{T})$ as follows. For $i=1,\dots,k$, consider a line $\ell_i$ defined as $y=2i-1$. Place each vertex $v \in V_i$, at the intersection between $\ell_i$ and the left side of $R(v)$. Note that such an intersection exists since the clique-planar representation of $(G(V,E),S,\psi)$ satisfies Property~1.
Draw each edge $(u,v)$ with $u \in V_i$ and $v \in V_{i+1}$ as a curve composed of three parts: The middle part coincides with the drawing of the corresponding link-edge, which connects a point $p_u$ on the top side of $R(u)$ with a point $p_v$ on the bottom side of $R(v)$; the first part is a curve connecting $u$ with $p_u$ entirely contained inside $R(u)$ not crossing any other edge (this can be done by routing the curve first following the left side of $R(u)$ and then following the top side of $R(u)$); and the last part is a curve connecting $v$ with $p_v$ entirely contained inside $R(v)$ not crossing any other edge (this can be done by routing the curve first following the left side of $R(v)$ and then following the bottom side of $R(v)$).

We show that $\Gamma$ is a $T$-level planar drawing of $(V,E',\gamma,\mathcal{T})$. No two edges cross in $\Gamma$ since the middle parts of the edges in $E'$ have the same drawing as the link-edges in $E$, which do not cross by hypothesis, while the first and the last parts do not cross by construction. Finally, the fact that the ordering of the vertices of $V_i$ along $\ell_i$ is compatible with $T_i$ descends from the fact that $\ell_i$ intersects all the rectangles of each clique $s$ with $\psi(s)=i$ and that no two rectangles representing vertices belonging to different cliques overlap. Hence, vertices belonging to the same clique, and hence children of the same internal node of $T_i$, are consecutive along $\ell_i$.

The construction can be performed in linear time, thus proving the lemma.
\end{proof}

We thus get the main result of this section. 

\begin{theorem}
{\sc Level Clique Planarity} is solvable in quadratic time for proper instances and in quartic time for general instances.
\end{theorem}

\begin{proof}
Any instance $(G,S,\psi)$ of {\sc Level Clique Planarity} can be made proper by introducing dummy cliques composed of single vertices to split link-edges spanning more than one level. This does not alter the level clique planarity of the instance and might introduce a quadratic number of vertices.
Lemma~\ref{le:hierarchical-2-cyclicTlevel} constructs in linear time an equivalent proper instance of {\sc T-level Planarity}. The statement follows since {\sc T-level Planarity}  can be solved in quadratic time~\cite{tibp-addfr-15} for proper instances.
\end{proof}

\section{Conclusions and Open Problems}\label{se:conclusions}

We initiated the study of hybrid representations of graphs in which vertices are geometric objects and edges are either represented by intersections (if part of dense subgraphs) or by curves (otherwise). Several intriguing questions arise from our research. 

\begin{enumerate}
\item How about considering families of dense graphs richer than cliques? Other natural families of dense graphs could be considered, say interval graphs, complete bipartite graphs, or triangle-free graphs. 
\item How about using different geometric objects for representing vertices? Even simple objects like equilateral triangles or unit circles seem to pose great challenges, as they give rise to arrangements with a complex combinatorial structure. For example, we have no counterpart of Lemma~\ref{le:squares-arrangment} in those cases. \item What is the complexity of the bipartite 2-page book embedding problem? We remark that, in the version in which spine crossings are allowed, this problem is equivalent to the clique planarity problem for instances with two cliques.
\end{enumerate}

\bibliographystyle{splncs03}
\bibliography{bibliography}

\begin{thebibliography}{10}
\providecommand{\url}[1]{\texttt{#1}}
\providecommand{\urlprefix}{URL }

\bibitem{adfjkpr-tppeg-j14}
Angelini, P., {Di Battista}, G., Frati, F., Jelinek, V., Kratochvil, J.,
  Patrignani, M., Rutter, I.: Testing planarity of partially embedded graphs.
  ACM Trans. on Algorithms  11(4) (2015)

\bibitem{tibp-addfr-15}
Angelini, P., {Da Lozzo}, G., {Di Battista}, G., Frati, F., Roselli, V.: The
  importance of being proper (in clustered-level planarity and {T}-level
  planarity). Theor. Comp. Sci.  571,  1--9 (2015)

\bibitem{bbdlpp-valg-11}
Batagelj, V., Brandenburg, F., Didimo, W., Liotta, G., Palladino, P.,
  Patrignani, M.: Visual analysis of large graphs using (x, y)-clustering and
  hybrid visualizations. {IEEE} Trans. Vis. Comput. Graph.  17(11),  1587--1598
  (2011)

\bibitem{br-spacep-13}
Bl{\"{a}}sius, T., Rutter, I.: Simultaneous {PQ}-ordering with applications to
  constrained embedding problems. In: Khanna, S. (ed.) {SODA} '13. pp.
  1030--1043. {SIAM} (2013)

\bibitem{brw-wnv-01}
Brandes, U., Raab, J., Wagner, D.: Exploratory network visualization:
  Simultaneous display of actor status and connections. Journal of Social
  Structure  2 (2001)

\bibitem{f-sslg-09}
Breu, H.: Algorithmic Aspects of Constrained Unit Disk Graphs. Ph.D. thesis,
  The University of British Columbia, Canada (1996)

\bibitem{cgp-mg-02}
Chen, Z., Grigni, M., Papadimitriou, C.H.: Map graphs. J. {ACM}  49(2),
  127--138 (2002)

\bibitem{cdfpp-cccg-06}
Cortese, P.F., {Di Battista}, G., Frati, F., Patrignani, M., Pizzonia, M.:
  C-planarity of c-connected clustered graphs. J. Graph Algorithms Appl.
  12(2),  225--262 (2008)

\bibitem{fce-pcg-95}
Feng, Q.W., Cohen, R.F., Eades, P.: Planarity for clustered graphs. In:
  Spirakis, P.G. (ed.) ESA '95. LNCS, vol. 979, pp. 213--226. Springer (1995)

\bibitem{hb-vosn-05}
Heer, J., Boyd, D.: Vizster: Visualizing online social networks. In: Stasko,
  J.T., Ward, M.O. (eds.) InfoVis '05, 23-25 Oct. 2005, Minneapolis, {USA}.
  p.~5. {IEEE} Computer Society (2005)

\bibitem{hfm-dhvsn-07}
Henry, N., Fekete, J., McGuffin, M.J.: Nodetrix: a hybrid visualization of
  social networks. {IEEE} Trans. Vis. Comput. Graph.  13(6),  1302--1309 (2007)

\bibitem{hn-sattpbep-14}
Hong, S., Nagamochi, H.: Simpler algorithms for testing two-page book embedding
  of partitioned graphs. In: Cai, Z., Zelikovsky, A., Bourgeois, A.G. (eds.)
  Computing and Combinatorics - 20th International Conference, {COCOON} 2014,
  Atlanta, GA, USA, August 4-6, 2014. Proceedings. LNCS, vol. 8591, pp.
  477--488. Springer (2014),
  \url{http://dx.doi.org/10.1007/978-3-319-08783-2_41}

\bibitem{i-isgci}
Irzhavsky, P.:
  {{I}}nformation~{S}ystem~on~{G}raph~{C}lasses~and~their~{I}nclusions~({I}{S}{G}{C}{I}).
  \url{http://graphclasses.org/classes/refs1600.html#ref_1660}

\bibitem{t-mgp-98}
Thorup, M.: Map graphs in polynomial time. In: {FOCS} '98. pp. 396--405. {IEEE}
  (1998)

\bibitem{w-chcpmpg-82}
Wigderson, A.: The complexity of the {H}amiltonian circuit problem for maximal
  planar graphs. EECS Department Report 298, Princeton University, 1982

\end{thebibliography}

\end{document}